\documentclass[10pt,twocolumn]{article}
\usepackage[top=1in, bottom=1in, left=0.85in, right=0.85in]{geometry}
\usepackage{amsmath,amsthm,amssymb}
\usepackage{mathtools}
\usepackage{graphicx}
\usepackage{hyperref}
\usepackage{mathrsfs}
\usepackage{physics}
\usepackage{dsfont}
\usepackage{xcolor}
\usepackage{authblk}
\usepackage{cuted}      
\usepackage{balance}    
\usepackage{listings}
\usepackage[backend=biber]{biblatex}
\usepackage{booktabs}
\usepackage{makecell}
\definecolor{Green}{rgb}{0,.8,.4}
\definecolor{Plum}{rgb}{.5,0,1}
\definecolor{cyan}{rgb}{0.50,.9,0.9}
\definecolor{lightblue}{rgb}{0.3, 0.6, .7}
\definecolor{ocre}{RGB}{243,102,25}
\definecolor{mygray}{RGB}{243,243,244}

\usepackage{fancyvrb}
\usepackage{bbm}

\newtheorem{thm}{Theorem}[section]

\newtheorem{lemma}[thm]{Lemma}

\theoremstyle{definition}

\theoremstyle{remark}

\numberwithin{equation}{section}

\newcommand{\caH}{{\mathcal H}}

\newcommand{\caL}{{\mathcal L}}

\newcommand{\caV}{{\mathcal V}}

\newcommand{\Z}{\mathbb{Z}}

\newcommand{\be}{\begin{equation}}
	\newcommand{\ee}{\end{equation}}
\newcommand{\bea}{\begin{eqnarray}}
	\newcommand{\eea}{\end{eqnarray}}
\newcommand{\beann}{\begin{eqnarray*}}
	\newcommand{\eeann}{\end{eqnarray*}}

\title{Matrix Product State Representatives of Symmetry-Protected Phases}
\author{Thomas Jackson}
\date{\today}

\begin{document}
\twocolumn[
  \begin{@twocolumnfalse}
    \maketitle
    \begin{abstract}
    We provide a simple ansatz for finding integer spin $SO(3)\times\Z_2^T\times\Z_2^R$-invariant matrix product states representing a phase for each triple of $\Z_2$ topological indices studied in the literature by Tasaki \cite{tasakiheistop}\cite{tasakitop} and Ogata (corresponding to time-reversal and reflection symmetries) \cite{ogatatime}\cite{ogatareflection}. Every combination of trivial and non-trivial indices has a spin-$1$ $SU(2)$, time-reversal, and site-reflection symmetric MPS representative with injectivity length $l\leq 8$. We provide some exact examples for spin-$1$ and a Fortran program which will find representatives of any triple for any integer spin, as well as some data from random sampling for small spin. The existence of such states agrees with the classification of Chen, Gu, and Wen \cite{chen}, and with the recent results of Tasaki \cite{tasakinew} equating the hidden-order and one of Ogata's indices, wherein the Ogata index is that associated to the dihedral $D_2$ symmetry.
    \vspace{2\baselineskip} 
    \end{abstract}
  \end{@twocolumnfalse}
]

\maketitle
\section{Introduction}
The classification of quantum phases has been an extremely fruitful and interdisciplinary area of research in recent years \cite{bachmann}\cite{chen}\cite{ogata}. Mathematical models for arbitrary phases in one-dimension have been somewhat \textit{ad hoc}. Here we provide a convenient matrix product state ansatz for phases distinguished by an on-site $SO(3)$ and global time-reversal $\Z_2$ and reflection $\Z_2$ symmetries. The Haldane phase is protected by each of these individually, however there exist phases of matter in all $8$ combinations of trivial and non-trivial indices \cite{chen}. We review the calculation of these indices before providing the ansatz; its form can be easily generalized to arbitrary on-site symmetry, although proving that it is injective is another step. In the case of spin-$1$ we provide an explicit example relying on the neat fact that the coefficients of the Clebsch-Gordan tensor lie in a particular ring of algebraic integers. We provide Fortran code to recreate this random sampling for any spin-$s$.
This result is mainly meant to communicate the flexibility of the ansatz, which can be seen as a generalization of the AKLT model \cite{aklt} with a degenerate bond algebra.\par
To clarify any confusion, the phases here with "flipped" time-reversal and reflection indices are not in contradiction to the recent result \cite{tasakinew}, in which the Ogata index is the one associated to the $D_2$ symmetry.
\section{Definitions and topological indices}
For now we will speak in terms of general spin $s$ with our onsite Hilbert space is $\caH_x\cong\mathbb{C}^{2s+1}$. We define our matrix product state tensor $A=(A^{-1},A^0, A^1)$ of bond dimension $D$, which defines a state on a finite chain $[1,N]$:
$$|\psi\rangle=\sum_{i_1,...,i_N}\Tr(A_{i_1}...A_{i_N})|i_1,...,i_N\rangle$$
and call it normal with injectivity length $l$ if $l$ is the least integer such that \begin{equation}
\mathrm{span}\{A^{i_1},...,A^{i_l}:\: i_k\in\{-1,0,1\}\}=\mathcal{M}_D(\mathbb{C})
\end{equation} meaning that after blocking $l$ sites this tensor is injective, which implies primitivity of the transfer operator $$\mathbb{E}(X)=\sum_i A^i X (A^i)^*$$ Such a tensor defines a pure infinite-volume matrix product state with a gapped, finite-range parent Hamiltonian \cite{fcs}.
There are three indices we consider, corresponding to on-site and global symmetries.

\subsection{Tasaki's hidden order index}
For an infinite-volume state $\omega:\mathcal{A}\to\mathbb{C}$, we use Tasaki's twist operator \cite{tasaki}
\begin{equation}
V_l=\bigotimes_{j=-l}^{l+1}\exp{\left[-2\pi i\frac{j+l}{2l+1}S^z_j\right]}
\end{equation}
where $S^z_j$ is the $z$ spin matrix at site $j$. It was showed that for $SU(2)$-invariant states the expectation of this operator is real and has a constant non-zero sign for all sufficients large $l$, meaning \begin{equation}
\lim_{l\to\infty}\frac{\omega(V_l)}{|\omega(V_l)|}=:\sigma_H\in\Z_2
\end{equation}This index $\sigma_H\in\Z_2$, distinguishes, in the case of matrix-product states, the type of representation of the bond algebra; it has been shown to be non-trivial for two states \cite{tasakitop}\cite{tasakiheistop} and under the stability assumptions of \cite{quasilocality2} has been shown to be non-trivial for a continuum of states in an infinite-dimensional (though small) pertubative region around the AKLT point, and, more broadly, for a perturbative region around any matrix product state with non-trivial Tasaki index $\sigma_H=-1$. For the matrix product states studied here we provide a short proof that $\omega(V_l)\to\sigma_H$. We will work through how to compute this index for the specific family of MPS studied here in a later section. Under $U(1)\rtimes\Z_2$ on-site symmetry, it was shown to be equivalent to Ogata's $D_2$ index \cite{tasakinew}.

\subsection{Ogata's time-reversal index}
First studied in \cite{ogatatime}, the index we label $\sigma_T\in\Z_2$ arises from the time-reversal anti-unitary:
\begin{equation}
\begin{split}
\Theta&=e^{-i\pi S_{\mathrm{tot}^y}}\mathcal{C},\quad S_{\mathrm{tot}}^y=\sum_{j\in\Z}S_j^y,\\
\Theta|i\rangle&=(-1)^{1-i}|-i\rangle, \quad i\in\{-1,0,1\}
\end{split}
\end{equation}
where $\mathcal{C}$ is the complex conjugation operator in the $S^z$ basis.
The index $\sigma_T$ is defined as the sign associated to the square of implementing time reversal on the (right) half-chain. For the matrix product states considered here this antiunitary $W$ acts as
where \begin{equation}
\begin{split}
W\overline{A^i}W^{-1}&=(-1)^{1-i}A^{-i},\\
W\overline{W}&=\sigma_T\mathbbm{1},\quad i\in\{-1,0,1\} 
\end{split}
\end{equation}
The first relation can contain a scalar phase but the states considered here do not have this property.

\subsection{Ogata's reflection index}
We label the index $\sigma_R\in\Z_2$ which can be found in \cite{ogatareflection}, that comes from the whole-chain reflection symmetry of the state. It is defined by \begin{equation}
(A^q)^T=pR^{-1}A^q R,\quad R^T=:\sigma_R R
\end{equation} where $R^2=R_x^2=\mathbbm{1}$ is the intertwiner enacting the standard global reflection at a site $x$, and $p\in\{-1,+1\}$. Here we assume in the first equality translation invariance.

We note that the recent result \cite{tasakinew} proves equivalence of Tasaki's twist operator with the index associated to the dihedral symmetry under on-site $U(1)\rtimes \mathbb{Z}_2$, which is separate from those of time-reversal and reflection.
The main result of this paper is the following theorem:
\begin{thm}
There exists an $SO(3)\times\Z_2^T\times\Z_2^R$-invariant, normal matrix product state $\psi$ for any triple of indices $(\sigma_H,\sigma_T,\sigma_R)\in\Z_2\times\Z_2\times\Z_2$ with injectivity length $l\leq 8$.
\end{thm} 
The proof follows the background and construction of the relevant MPS tensors.

\section{Construction of the Family}
The MPS tensors are chosen such that the $SU(2)$-symmetry is apparent; the main difference between this set of MPS and many studied before is the size of the multiplicity space; its is this that allows freedom to independently choose the topological indices. Let $\alpha\neq\beta$ be irreducible representations of $SU(2)$, and let our spin chain have onsite physical representation $D^{(s)}$ with $s\in\Z$, with \begin{equation}\label{containment}
\begin{split}
D^{(s)}&\subset\alpha\otimes\beta,\quad D^{(s)}\subset\alpha\otimes\alpha,\\
D^{(s)}&\subset\beta\otimes\beta
\end{split}
\end{equation}
or numerically \begin{equation}|\beta-\alpha|\leq s\leq \min(2\alpha,2\beta)\end{equation}
Note \eqref{containment} implies $\alpha-\beta\in\Z$.
Let $C^i_{\gamma,\gamma'}$ be the conventional $SU(2)$ Clebsch-Gordan coefficients finding the spin-$s$ basis vectors $i$ inside the tensor product $\gamma\otimes\gamma'$, meaning:
\begin{equation}\label{cg}
(C^i_{\gamma,\gamma'})_{m,n}=\langle\gamma, n; s,i|\gamma,m\rangle
\end{equation}Then we define our MPS tensors
\begin{equation}\label{tensors}
A^i=\begin{bmatrix}
C^i_{\alpha,\alpha}\otimes M_{\alpha,\alpha} & C^i_{\alpha,\beta} \otimes M_{\alpha,\beta} \\
C^i_{\beta,\alpha} \otimes M_{\beta,\alpha}& C^i_{\beta,\beta}\otimes M_{\beta,\beta}
\end{bmatrix}
\end{equation}
where $M_{\gamma,\gamma'}\in \mathcal{M}_n(\mathbb{C})$ for some $n\geq 0$ and $i$ are the canonical basis for $\mathbb{C}^{2s+1}$.
Some comments on the motivation of the details in the construction are in order. This originated as a search for a nice AKLT-like state with bond algebra isomorphic to $$\caV\cong n\caV_{\frac{1}{2}}\oplus n\caV_{\frac{3}{2}}$$ for $n>1$ with constrained topological indices (i.e. within the Haldane phase). In the process of such a construction, the topological indices computed showed that the Ogata time-reversal index and Tasaki indices can be different in this ansatz. A quick check found that enlarging the multiplicities to \begin{equation}\caV\cong 4\caV_{\frac{1}{2}}\oplus 4\caV_{\frac{3}{2}}\cong (\caV_{\frac{1}{2}}\oplus \caV_{\frac{3}{2}})\otimes \mathbb{C}^4\end{equation} finds gapped ground states with any given triple \begin{equation}(\sigma_H,\sigma_T,\sigma_R)\in\Z_2\times \Z_2\times \Z_2\end{equation}
of topological indices.\par
The nature of the above construction can be viewed as the following: each matrix $M_{\gamma,\gamma'}$ chooses a particular pair of bond representations $\gamma,\gamma'$ from the left- and right-hand virtual bonds and changes how much the basis vector term $i$ is influenced by the spin-$s$ sector of their tensor product. Matrix product states allow for not just superpositions of such subrepresentations (the case of $n=1$), but of nontrivial algebraic relations between these subrepresentations; most importantly, these algebraic relations can produce gapped states while preserving the overall physical and virtual symmetries.\par
Constraining the topological indices becomes a series of constraints on the MPS tensor and thus the set of multiplicity matrices $(M_{ij})$. In order to define these constraints, we let \begin{equation}
\begin{split}
(U_j)_{m,n}&=(-1)^{j-m}\delta_{m,-n},\quad U_j^T=(-1)^{2j}U_j,\\
U_J\overline{U}_j&=(-1)^{2j}\mathbbm{1},\quad U_jU_j^T=\mathbbm{1}.
\end{split}
\end{equation}
The convention \eqref{cg} implies \begin{equation}
\label{firstcon}
U_{\gamma'}\overline{C^i_{\gamma',\gamma}}U_{\gamma}^{-1}=(-1)^{1-q}C_{\gamma,\gamma'}^{-i}
\end{equation}
\begin{equation}
\label{secondcon}
U_{\gamma'}^{-1}C^i_{\gamma',\gamma}U_{\gamma}=(-1)^{\gamma'-\gamma-1}\sqrt{\frac{2\gamma'+1}{2\gamma+1}}\left(C_{\gamma,\gamma'}^{i}\right)^T
\end{equation}
Now let $J,K$ be as-yet undefined real orthogonal $4\times 4$ matrices satisfying \begin{equation}
J\overline{J}=\sigma_J\mathbbm{1},\quad K^T=\sigma_K K,\quad JK=KJ
\end{equation}
and define \begin{equation}
\begin{split}
W&=\mathrm{diag}(U_{\alpha}\otimes J,U_{\beta}\otimes J),\\
R&=\mathrm{diag}\left(U_{\alpha}\otimes K, U_{\beta}\otimes K\right).
\end{split}
\end{equation}
The point of these definitions is to allow the time-reversal and reflection intertwiners to "miss" the $SU(2)$-invariance of the state and similarly "miss" each other in the multiplicity space, meaning the sign of the reflection and time-reversal indices, which would naturally arise from the Clebsch-Gordan intertwiner, is "undone" in the multiplicity space.\par
Thus we are looking for multiplicity matrices satsifying the constraints \begin{equation}
\begin{split}
J\overline{M_{\gamma',\gamma}}J^{-1}&=M_{\gamma',\gamma},\quad M^T_{\gamma,\gamma}=-p K^{-1}M_{\gamma,\gamma}K,\\
M_{\beta,\alpha}&=p\sqrt{\frac{2\alpha+1}{2\beta+1}}KM_{\alpha,\beta}^TK^{-1}.
\end{split}
\end{equation}
The following lemma then shows that \begin{equation}
\begin{split}
\sigma_H&=(-1)^{2\alpha}=(-1)^{2\beta},\\
\sigma_T&=\sigma_H\sigma_J,\quad \sigma_R=\sigma_H\sigma_K
\end{split}
\end{equation}
\begin{lemma} Let $\sigma_H$ be the hidden-order index defined above for the MPS tensor associated to the spin $(\alpha,\beta)$ defining a normal MPS; then
\begin{equation}\sigma_H=(-1)^{2\alpha}=(-1)^{2\beta}\end{equation}
\end{lemma}
\begin{proof}
For each choice of $\alpha,\beta$ let $\mathbb{E}(X)$ be the primitive normalized transfer operator, and let $R,L$ be its respective left and right fixed points. Let $$Z=\mathrm{diag}(S_{\alpha}^z\otimes \mathbbm{1}_4,S_{\beta}^z\otimes\mathbbm{1}_4),\quad V(\theta)=e^{-i\theta Z}$$
then \begin{equation}
e^{-i\theta i}A^i=V(\theta)A^iV(\theta)^{-1},\quad V(2\pi)=(-1)^{2\alpha}=(-1)^{2\beta}
\end{equation}
Defining $\caL_V(X)=VX$, the transfer operator with a twist is $\mathbb{E}_{\theta}=\mathcal{L}_{V(\theta)}\mathbb{E}\mathcal{L}_{V(\theta)^{-1}}$. On sites $0,...,N$ with $N=2l+1$ with $\Delta=\frac{2\pi}{N}$ we have
\begin{equation}
\omega(V_l)=(-1)^{2\alpha}\mathrm{tr}\left[L(\mathbb{E}\mathcal{L}_{V(\Delta)})^N(R)\right]
\end{equation}
The leading simple eigenvalue $\lambda(\Delta)$ of $\mathbb{E}\mathcal{L}_{V(\Delta)}$ has the properties \begin{equation}
\lambda(0)=1,\quad \lambda'(0)=-i\mathrm{tr}(LZR)=0
\end{equation}
with the last equality following from $z$-rotation invariance of the fixed point, and so $\lambda(\Delta)^N\to 1$ and thus $\omega(V_l)\to\sigma_H=(-1)^{2\alpha}=(-1)^{2\beta}$.
\end{proof} 
Note the $z$-rotation invariance is also crucial to showing the Ogata $D_2$ index is equivalent to the Tasaki index \cite{tasakinew}.
\subsection{Finding representatives of each phase} 
From here we will assume $s=1$.
An explicit MPS tensor satisfying the constraints above can be found by simple random sampling of the subspace given by the constraints above and the $J$ and $K$ matrices defined in the proof below. A cute explicit example can be found which relies on the fact that the coefficients of $C^{i}_{\gamma,\gamma'}$ are contained in a particular ring of algebraic integers. We provide the derivation here for the curious, and constitutes a proof of Theorem 1 via computation.
\begin{proof}
Define \begin{equation}j=\begin{bmatrix}
0 & 1 \\ -1 & 0
\end{bmatrix}, \quad p=-\sigma_H\sigma_R
\end{equation}
\begin{equation}
\begin{split}
J&=\begin{cases}
\mathbbm{1}_4,& \sigma_H\sigma_T=1,\\
\mathrm{diag}(j,j), & \sigma_H\sigma_T=-1
\end{cases},\\
K&=\begin{cases}
\mathrm{diag}(1,1,-1,-1),& \sigma_H\sigma_R=1,\\
\mathrm{diag}(j,-j), & \sigma_H\sigma_R=-1
\end{cases}
\end{split}
\end{equation}
We note $JK=KJ$.
The entries of $A$ lie in $\Z[\frac{1}{30},i,\sqrt{2},\sqrt{3},\sqrt{5},\sqrt{7}]$ with denominators not divisible by the prime $1009$ which is the smallest prime $p$ where all the prime roots have a integer solution $\mod p$. This fact can be worked out from the rational Racah formula, the proof of which we elide, as one can also square each term in the Clebsch-Gordan tensors to see this. We map the matrices under the homomorphism to $\mathbb{F}_{1009}$ given by \begin{equation}
\begin{split}
i&\mapsto 469, \quad \sqrt{2}\mapsto 439,\quad \sqrt{3}\mapsto 149,\\
\sqrt{5}&\mapsto 244, \quad \sqrt{7}\mapsto 45
\end{split}
\end{equation}
We choose random matrices with coefficients in the set $\{-3,-2,...,2,3; -3i,-2i,...,2i,3i\}$:
\begin{align*}
 Z_1&=\begin{pmatrix}2&3&3&0\\3&3&3&-3\\0&1&-2&-1\\1&2&1&-2\end{pmatrix}
 +i\begin{pmatrix}1&3&-2&0\\-1&3&-3&0\\3&0&-3&2\\3&3&3&-1\end{pmatrix},\\[4pt]
 Z_2&=\begin{pmatrix}2&3&1&3\\1&-2&0&1\\-1&1&-1&3\\-2&1&-3&-3\end{pmatrix}
 +i\begin{pmatrix}0&0&0&0\\2&0&-1&3\\3&-1&-2&-2\\-2&0&1&1\end{pmatrix},\\[4pt]
 Z_3&=\begin{pmatrix}3&3&0&1\\-3&-2&3&3\\-3&0&0&1\\-3&0&-1&-2\end{pmatrix}
 +i\begin{pmatrix}2&1&-2&2\\0&-2&1&0\\0&-1&3&-2\\3&-2&2&3\end{pmatrix}.
\end{align*}and let \begin{equation}
X_k=Z_k+J\overline{Z}_kJ^{-1},\quad k=1,2,3
\end{equation}
Then an explicit injective solution is given by defining the multiplicity matrices:
\begin{equation}
\begin{split}
M_{\alpha,\alpha}&=X_1-pKX_1^TK^{-1},\\
M_{\beta,\beta}&=X_2-pKX_2^TK^{-1}
\end{split}
\end{equation}
\begin{equation}
\begin{split}
M_{\alpha,\beta}&=X_3,\\
M_{\beta,\alpha}&=p\sqrt{\frac{2\alpha+1}{2\beta+1}}KX_3^TK^{-1}
\end{split}
\end{equation} and $C^i_{\alpha,\beta}$ etc. determined by the desired $\sigma_H$.
\end{proof}
The Fortran code at the GitHub link evaluates the rational Racah formula for the Clebsch-Gordan matrices by mapping to $\mathbb{F}_{1009}$ and computes via row reduction $\mathrm{span}\{A^{i_1},...,A^{i_l}:\: i_k\in\{-1,0,1\}\}$ using the above random matrices until injectivity is achieved.

\begin{table*}
\centering
\small
\setlength{\tabcolsep}{6pt} 
\caption{States with lowest correlation length from sampling 100,000 per topological sector for spin-$1$}
\label{t:table}
\begin{tabular}{|c|c|c|c|c|}
\hline\hline
\makecell{$(\sigma_H,\sigma_T,\sigma_R)$} & \makecell{Bond\\dimension} & \makecell{Injectivity\\length} & \makecell{Correlation\\length} \\
\hline
$(-1,-1,-1)$ & 12 & 6 &  0.7191  \\
\hline
$(+1,+1,+1)$ & 16 & 6 & 0.9500 \\
\hline
$(-1,+1,-1)$ & 12 & 6 & 1.9814 \\
\hline
$(+1,-1,+1)$ & 16 & 6 &  2.1225 \\
\hline
$(+1,-1,-1)$ & 32 & 8 & 4.4197\\
\hline
$(+1,+1,-1)$ & 32 & 8 & 4.4879 \\
\hline
$(-1,+1,+1)$ & 24 & 8 & 4.5598 \\
\hline
$(-1,-1,+1)$ & 24 & 8 & 4.7874 \\
\hline
\end{tabular}
\end{table*}
\begin{table*}
\centering
\small
\setlength{\tabcolsep}{6pt}
\caption{States with lowest correlation length from sampling 10,000 per topological sector for spin-$2$}
\label{t:table_new2}
\begin{tabular}{|c|c|c|c|}
\hline\hline
\makecell{$(\sigma_H,\sigma_T,\sigma_R)$} & \makecell{Bond\\dimension} & \makecell{Injectivity\\length} & \makecell{Correlation\\length} \\
\hline
$(+1,+1,+1)$ & 16 & 5 & 0.9031 \\
\hline
$(-1,-1,-1)$ & 20 & 5 & 1.0706 \\
\hline
$(+1,-1,+1)$ & 16 & 5 & 1.8342 \\
\hline
$(-1,+1,-1)$ & 20 & 5 & 1.9765 \\
\hline
$(-1,+1,+1)$ & 40 & 7 & 4.8291 \\
\hline
$(+1,-1,-1)$ & 32 & 7 &  4.8949 \\
\hline
$(+1,+1,-1)$ & 32 & 7 & 5.1332 \\
\hline
$(-1,-1,+1)$ & 40 & 7 & 5.1394 \\
\hline
\end{tabular}
\end{table*}
\begin{table*}
\centering
\small
\setlength{\tabcolsep}{6pt}
\caption{States with lowest correlation length from sampling 1,000 per topological sector for spin-$3$}
\label{t:table_new}
\begin{tabular}{|c|c|c|c|}
\hline\hline
\makecell{$(\sigma_H,\sigma_T,\sigma_R)$} & \makecell{Bond\\dimension} & \makecell{Injectivity\\length} & \makecell{Correlation\\length} \\
\hline
$(-1,-1,-1)$ & 20 & 5 & 1.0120 \\
\hline
$(+1,+1,+1)$ & 24 & 5 & 1.3313 \\
\hline
$(+1,-1,+1)$ & 48 & 6 &  1.9135\\
\hline
$(-1,+1,-1)$ & 40 & 6 &  1.8656\\
\hline
$(-1,+1,+1)$ & 40 & 7 & 4.5967 \\
\hline
$(+1,-1,-1)$ & 48 & 7 & 4.6122  \\
\hline
$(+1,+1,-1)$ & 48 & 7 & 5.1620 \\
\hline
$(-1,-1,+1)$ & 40 & 7 & 5.1394 \\
\hline
\end{tabular}
\end{table*}
The injectivity lengths listed are generic in the set of MPS tensors, as is the range of the parent Hamiltonian, via the arguments of \cite{smpg}, although shorter injectivity length representatives of each phase can be found, e.g. the AKLT state for $(-1,-1,-1)$, which we note has longer correlation length than the minimal representative of the Haldane phase $(-1,-1,-1)$ found in the table.

\section{Conclusion and Outlook}
This paper is meant mainly to provide a simple matrix product state ansatz for those looking for phase representatives.
We can see that the tensors provided in \eqref{tensors} can be generalized \textit{mutatis mutandis} to other on-site symmetry groups, although the existence of a normal solution generically is seemingly more of an algebraic-geometric statement. We can see that a corollary of Theorem 1 via \cite{smpg} is that this injectivity length and parent Hamiltonian range is generic within the subspace described. At the moment proving that a candidate always exists for spin-$s$ of a particular injectivity length seems difficult, and it would be interesting to see if any of the "flipped" phases exist in the nearest-neighbor $SU(2)$-invariant phase diagram for some integer $s$. Some indication that the "simplest" triples $(+1,+1,+1)$ and $(-1,-1,-1)$ are easier to realize physically is given by the trend in correlation lengths, although it seems unclear how to even rigorously characterize this. The generalization of the ansatz to tensor network states in two or higher dimensions formally does not seem too difficult, but showing any particular choice defines even a unique state is at present hard. 
\section{AI and Code Statement}
The main ideas of this paper are human as are its mistakes. The ansatz became convenient while exploring MPS approximants to $SU(2)$-symmetric states using among other things a healthy amount of freely available (Claude and Gemini, various versions) coding agents and chatbots over the past few months; these free tools were used for some of the formatting of equations and tables in the TeX. The author used whatever paid ChatGPT model is current to proofread it, and it supplied the integer solution (the previous work was floating-point) via the ring of algebraic integers mapping, and the Fortran code verifying it, unprompted. This was sweet of it to do, so the author included it in the paper. The free LLMs retooled this code to work for general spin and multiplicity. The author takes full responsibility for the content and correctness of the results presented. The \texttt{indices7.f90} file for generating any number of random examples of every triple for any integer spin $s$ and multiplicity $n$ is given at the GitHub link: \url{https://github.com/ThomasAJackson/ltqo_code}. It is also commented into this TeX file. It can be compiled by installing \texttt{gfortran} and following the compilation instructions in the program comments.

\section{Acknowledgements}
The author would like to thank Martin Fraas and Bruno Nachtergaele for general guidance, and Can G\"{o}rmez for lending his tokens for the final proofread. This work was supported in part by the National Science Foundation under grant DMS-2510824.
\printbibliography
\end{document}